\documentclass[smallextended]{svjour3}
\usepackage[dvipsnames]{color}
\usepackage{xcolor}
\usepackage{epstopdf}
\usepackage{graphicx}
\usepackage{rotating}
\usepackage{cite}
\usepackage{url}
\usepackage{times}
\usepackage{caption}
\usepackage{multicol}
\usepackage[linesnumbered,ruled,vlined]{algorithm2e}
\usepackage{eqparbox,array}
\usepackage{amsmath}
\usepackage{todonotes}
\usepackage{mathtools}
\usepackage{multirow}
\usepackage{amssymb}
\usepackage{array,booktabs}
\usepackage{subfig}
\usepackage{tabularx}
\usepackage{textcomp}
\usepackage{enumitem}
\usepackage{hyperref}

\usepackage{amsthm}

\DeclareCaptionFont{tiny}{\tiny}
\DeclareCaptionFont{footnotesize}{\small}

\newcommand{\head}[1]{\noindent \textbf{#1}}

\DeclarePairedDelimiterX\braket[2]{\langle}{\rangle}{#1 \delimsize\vert #2}

\begin{document}

\title{Depth-Optimal Quantum Circuit Placement for Arbitrary Topologies}
\author{Debjyoti Bhattacharjee \and Anupam Chattopadhyay}

\institute{Debjyoti Bhattacharjee \at
	      Hardware and Embedded Systems Laboratory,\\
             School of Computer Science and Engineering,\\
Nanyang Technological University, Singapore\\
             \email{debjyoti001@ntu.edu.sg}           
           \and
           Anupam Chattopadhyay \at
	    School of Computer Science and Engineering,\\
Nanyang Technological University, Singapore\\
             \email{anupam@ntu.edu.sg}
}

\date{Received: date / Accepted: date}
\maketitle

\begin{abstract}
A significant hurdle towards realization of practical and scalable quantum computing is to protect the quantum states from inherent noises during the computation. In physical implementation of quantum circuits, a long-distance interaction between two qubits is undesirable since, it can be interpreted as a noise. Therefore, multiple quantum technologies and quantum error correcting codes strongly require the interacting qubits to be arranged in a nearest neighbor~(NN) fashion. The current literature on converting a given quantum circuit to an NN-arranged one mainly considered chained qubit topologies or Linear Nearest Neighbor~(LNN) topology. However, practical quantum circuit realizations, such as Nuclear Magnetic Resonance~(NMR), may not have an LNN topology. To address this gap, we consider an arbitrary qubit topology. We present an Integer Linear Programming~(ILP) formulation for achieving minimal logical depth while guaranteeing the nearest neighbor arrangement between the interacting qubits. We substantiate our claim with studies on diverse network topologies and prominent quantum circuit benchmarks.
\end{abstract}

\section{Introduction}\label{sec:intro}
Quantum computation~\cite{NC:2000} promises to expand the reach of computing beyond classical --- both theoretically and practically. In quantum computing, the operations take place on so called Qubits, which is a linear combination of the conventional Boolean states in the two dimensional complex Hilbert space. Each operation on these qubits can be defined by a unitary matrix~\cite{NC:2000} which is represented by means of quantum gates. A \emph{quantum gate} over the inputs \mbox{$X=\{x_1,\dots , x_n\}$} consists of a single target line~$t\in X$ and, one or more control line(s)~$c\in X$ with~\mbox{$t\neq c$}. The following gates define the commonly used quantum gate library.

\emph{NOT gate}: The qubit on the target line~$t$ is inverted.

\emph{Controlled NOT gate} (CNOT): The target qubit~$t$ is inverted if the control qubit~$c$ is 1. This gate belongs to the general class of Toffoli gates, when accomodating larger number of control qubits.

\emph{Controlled $V$ gate}: The $V$ operation is performed on the target qubit~$t$ if the control qubit~$c$ is 1. The $V$ operation is also known as the square root of NOT, since two consecutive $V$ operations are equivalent to an inversion.

\emph{Controlled $V^\dagger$ gate}: The $V^\dagger$ gate performs the inverse operation of the $V$ gate, i.e.~$V^{\dagger}=V^{-1}$.

\emph{SWAP gate}: The SWAP gate, as the name suggests, exchanges two qubits. This gate belongs to the general class of Fredkin gates, when accommodating control qubits.

A major challenge towards the realization of practical and scalable quantum computing is to achieve quantum error correction~\cite{Cross:2009:CCS:2011814.2011815}. Long-distance interacting Qubits is particularly susceptible to the noise. Therefore, prominent quantum technologies and quantum error correction codes, e.g. surface codes~\cite{2012PhRvA..86c2324F} require that the quantum gates must be formed with a nearest neighbour interaction. In the resulting circuits, the interacting Qubits may form a chain, as in a 1D Qubit layout, and therefore, these circuits are referred to as Linear Nearest Neighbor (LNN) circuits. Conversion of a quantum circuit to an LNN one can be achieved by using SWAP gates.

These SWAP gates allow for making all control lines and target lines adjacent and, by this, help to convert a given quantum circuit to a nearest neighbor one. More precisely, a cascade of adjacent SWAP gates can be inserted in front of each gate $g$ with non-adjacent circuit lines in order to shift the control line of $g$ towards the target line, or vice versa, until they are adjacent. This is shown using the following example.

\begin{example}\label{example:naive}
Consider the circuit depicted in Fig.~\ref{fig:example_naive_qua}. As can be seen, gates~$g_1$, $g_4$, and $g_5$ are non-adjacent. Thus, in order to make this circuit nearest neighbor compliant, SWAP gates in front and after all these gates are inserted as shown in Fig.~\ref{fig:example_naive_qua_nn}.
\end{example}
\begin{figure}
 \begin{minipage}{0.4\textwidth}
  \includegraphics[height=0.75in]{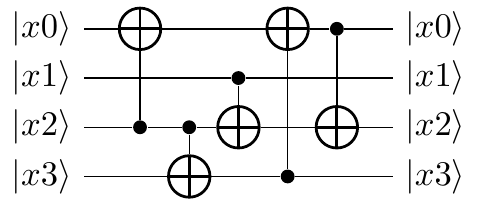}
  \caption{Given Circuit}
  \label{fig:example_naive_qua}
 \end{minipage}
\begin{minipage}{0.6\textwidth}
  \includegraphics[height=0.75in]{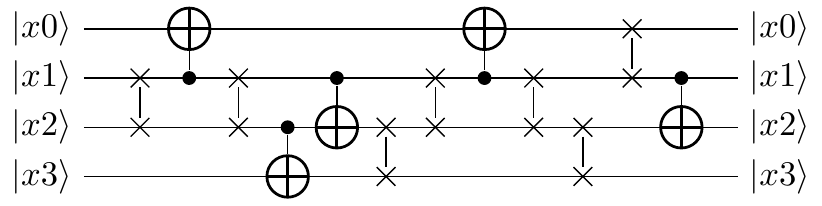}
  \caption{Nearest Neighbour Compliant Circuit}
  \label{fig:example_naive_qua_nn}
 \end{minipage}
\end{figure}

Quite a few works have been done in recent past to convert a quantum circuit to an LNN one by introducing additional swap gates, which, naturally impact the circuit performance by increasing the logical depth and gate count. To that effect, heuristic\cite{wille_aspdac16,mazdar_nnmmd16,Alireza2014,Rahman:2014:AQT:2711453.2629537,amlan_graph_partitioning,maslov_placement} and exact~\cite{wille_exact_nn} solutions are proposed, which balance the LNN conversion with other performance metrics. It is pointed out in~\cite{amlan_graph_partitioning} that the problem of nearest neighbour quantum circuit construction is equivalent to an NP-complete problem. Hence, it is unlikely that this problem can be solved optimally for large instances.

In parallel to the previous works, efficient LNN circuit construction has been studied for important quantum benchmarks, such as, quantum error correction~\cite{PhysRevA.69.042314} for Clifford+T gates~\cite{biswal}. In this work, we are primarily interested in the automated flow and for generic quantum circuits.

\subsection{Qubit Topology}\label{ssec:qubit_topo}
As noted in~\cite{maslov_placement}, the qubit topologies, on which the quantum circuit is to be mapped, are not necessarily of LNN structure. We provide a few examples here.

Recently, quantum error detection code is demonstrated on a square lattice~\cite{corcoles2015demonstration}. It also highlights the fact that for a classical bit-flip, linear array of qubits suffices, while for general fault detection, extending to higher-dimensional lattice structures is needed. 

Nuclear Magnetic Resonance (NMR) quantum computing achieved early success with realization of Shor's factorization algorithm~\cite{shor1999polynomial}. Liquid state NMR quantum computing utilizes the atomic spin states to realize the qubit and hence, has the molecular structures as qubit topologies. Solid state NMR has been also demonstrated~\cite{kampermann2002} using crystal of $NaNO_3$, essentially leading to molecular topologies.

A recent proposition for scalable quantum computer indicates that multiple, parallel quantum gates can be formed between distant qubits by controlling the lasers on Trapped Atomic Ions~\cite{brown2016co}. 

Harnessing atomic spins in endohedral fullerene molecules as qubits have also been reported~\cite{fullerene_quantum}. It has been further argued that molecular structures serve as a natural candidate for quantum technology by holding superpositions for longer period and ability to scaffold multiple molecules in a larger array.

Hence, an automated algorithm for achieving nearest neighbour interactions for a given quantum circuit while mapping on diverse qubit topologies are of significant practical interest. This is the main focus of current paper.

\subsection{Related Works}
To the best of our knowledge,~\cite{maslov_placement} and~\cite{whitney_grid_07} were the first to look into arbitrary topologies for quantum circuits with nearest neighbour constraints. So far, most of the other works in this domain have concentrated on 1D qubit layout or 2D qubit lattice structures~\cite{wille_aspdac16,Alireza2014}.

The work presented in~\cite{whitney_grid_07} focuses on identifying the qubit topology best suited for a given quantum circuit placement. In contrast, our focus is towards evaluating a given qubit topology and performing mapping on it. This particular problem has been dealt with in~\cite{maslov_placement} with examples taken from liquid state NMR molecules as the topologies. There, a graph partitioning-based approach is proposed and it is claimed to be asymptotically optimal for the case of chain nearest neighbour architecture. We address the same problem, by formulating it as an instance of ILP and show that optimal results are achievable for a wide variety of benchmarks and different topologies.

Independently, efficient qubit topology identification and the mapping flows for specific interaction graphs have been done in~\cite{beals2013efficient,brierley_butterfly}. For example, it is proved that for cyclic butterfly topology, the depth overhead for mapping a given quantum gate to a nearest neighbour one is $6\log n$. Subsequently, the mapping algorithm is also derived.

Communication and computation over networks is of major interest in quantum networks~\cite{quantum_butterfly} as well as for classical telecommunication networks. The problem of permutation routing on variety of graphs has been studied in the past~\cite{habermann1972parallel,spanke1987n,sau2007optimal}.

\subsection{Motivation and Contribution}
Despite the presence of diverse qubit topologies and need for an automated mapping flow of quantum circuits to such topologies, the current literature focuses mostly on 1D chain qubit and 2D lattice structures. 

\begin{itemize}
\item In order to address this gap, we present an ILP-based algorithm to realize depth-optimal nearest neighbour quantum circuits. Our algorithm is also applicable, naturally, to simpler structures. 

\item We benchmark the algorithm for diverse topologies and quantum circuits and compare the scalability and performance against previous exact NN optimization approaches.
\end{itemize}
\section{Preliminaries and problem statement}\label{sec:prelim}
\noindent In this section, we introduce the notations and terminologies for formally defining the nearest-neighbor optimization problem of quantum computing. Thereafter, we present three variants of the problem.

\begin{definition}
A \textbf{quantum circuit}, defined over $n$-qubits $q_1$, $q_2$,...,$q_n$ is a series of levels $L_i$, where each level $L_i$ consists of a set of quantum gates $G_i^1$, $G_i^2$, $\cdots$, $G_i^k$ with each gate $G_i^j$ operating on one or more qubits. Any two pair of gates $G_i^j$ and $G_i^k$ in a level $L_i$ do not operate on any common qubit and therefore can be executed in parallel. We assume that each level $L_i$ takes one cycle to execute. A quantum circuit with $k$ levels has a delay of $k$ cycles.
\end{definition}
\begin{figure}[!htb]
    \centering
    \begin{minipage}{2.5in}
        \centering
        \includegraphics[height=1.2in]{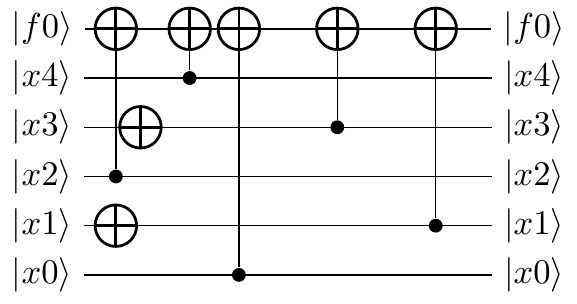}
        \caption{\em A quantum circuit\protect\footnotemark with delay 5}
        \label{fig:qc}
    \end{minipage}%
    \begin{minipage}{2.5in}
        \centering
        \includegraphics[height=1.2in]{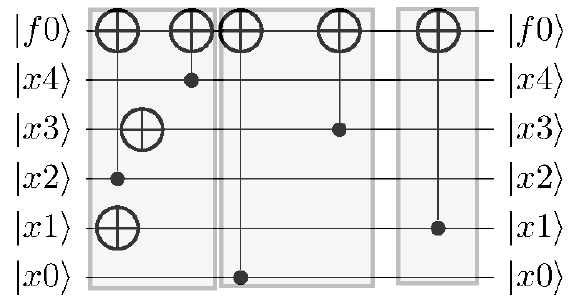}
        \caption{\em Interactions for block size b = 2}
        \label{fig:qc_block}
    \end{minipage}
\end{figure}
\footnotetext{xor5\_254.real file from RevLib}

  Given a quantum gate with $m$-control lines $l_1, ... , l_m$ and target line $l_t$, qubits $q_l$ and $q_{t}$ have to be  nearest-neighbors, $ 1 \le l \le m$. For level $L_i$,
  we define \textbf{interaction} $I_i$ as the set of nearest neighbors for the all the gates in $L_i$.   The levels and corresponding interactions of a quantum circuit is determined using Algorithm~\ref{algo:level}.

\begin{example}
 Fig.~\ref{fig:qc} shows a quantum circuit with 5 two-input Toffoli gates and 2 CNOT gates. The circuit has 5 levels and hence has a delay of 5.
 \begin{itemize}
  \item[$L_1$] : {\tt \small [t2 x2 f0, t1 x1, t1 x3]}
  \item[$L_2$] : {\tt \small [t2 x4 f0]}
   \item[$L_3$] : {\tt \small  [t2 x0 f0]}
 \item[$L_4$] : {\tt \small [t2 x3 f0]}
   \item[$L_5$] : {\tt \small [t2 x1 f0]}
 \end{itemize}
Corresponding to level $L_1$, interaction $I_1$ is [{\tt(x2, f0), (x1), (x3)}]. Similarly, $I_2, I_3, I_4$ and $I_5$ is [{\tt (x4,f0)}], [{\tt(x0,f0)}], [{\tt (x3,f0)}] and [({\tt x1,f0})] respectively.
\end{example}

\begin{algorithm}
\SetKwFunction{proc}{ComputeLevel}
 \SetKwProg{myproc}{Procedure}{}{QCkt}
  \myproc{\proc{devUseTable}}{
 levelList = []\;
 processedGate = set()\;
 L = set()\;
 Lvar = set()\;
\For { $G_i$ $\in$ QCkt}{
  \If { $G_i$ $\notin$ processedGate}{
      \If {$G_i.var \bigcap Lvar == \phi$}
      {
	L.add($G_i$)\;
	Lvar.add($G_i.var$)\;
	processedGate.add($G_i$)\;
	\For{$G_j$ $\in$ QCkt}{
	  \If {$G_j.var \bigcap Lvar == \phi$}
	  {
	    L.add($G_j$)\;
	    Lvar.add($G_j.var$)\;
	    processedGate.add($G_j$)
	  }
      }
      levelList.add(L)\;
      L = set()\;
      Lvar = set()\;
    } 
 }
} 
}
\textbf{return} reassignMap\;
\caption{Level Computation Algorithm}\label{algo:level}
\end{algorithm}

\begin{table}[h]
\centering
\caption{Minimum number of nodes in the smallest graph of each topology }
\label{table:minnodes}
\begin{tabular}{ll}\bottomrule
\textbf{Topology} & \textbf{Min. \#Nodes} \\ \midrule
 1D & 2 \\
 Cycle & 3 \\
 2D-Mesh & 9 \\
 Torus & 9 \\
 3D-Grid & 8 \\
 Cyclic butterfly & 24 \\\toprule
\end{tabular}
\end{table}
Physically, qubits can be arranged in various topologies, as discussed in the subsection~\ref{ssec:qubit_topo}. Such topologies allow interaction between only between some pairs of qubit positions. We introduce this constraint in the form of a topology graph.

\begin{definition}
A \textbf{topology graph} is an ordered pair  $T$=($T_\mathcal{V},T_\mathcal{E}$).  $T_\mathcal{V}$  is the vertex set, where each vertex $v \in T_\mathcal{V}$ represents a physical location where one qubit can reside. $T_\mathcal{E}$ is the edge-set, which contains a set of edges. An edge $e_{vw} \in T_\mathcal{E}$ indicates that qubit at location/vertex $v$ and $w$ can interact. In other words, qubits at location $v$ and $w$ are nearest-neighbors~(NN).
\end{definition}

\noindent Fig.~\ref{fig:topo} presents various topologies. The minimum number of nodes for the smallest graph of each topology is presented in Table~\ref{table:minnodes}. Given a quantum circuit with $n$-qubits, and a specific topology, we use the smallest topology graph $T$ such that $T_\mathcal{V} \ge n$ for realizing the quantum circuit.

\begin{figure}[b]
  \vspace{-0.2cm}
 \centering
 \includegraphics[width=4in]{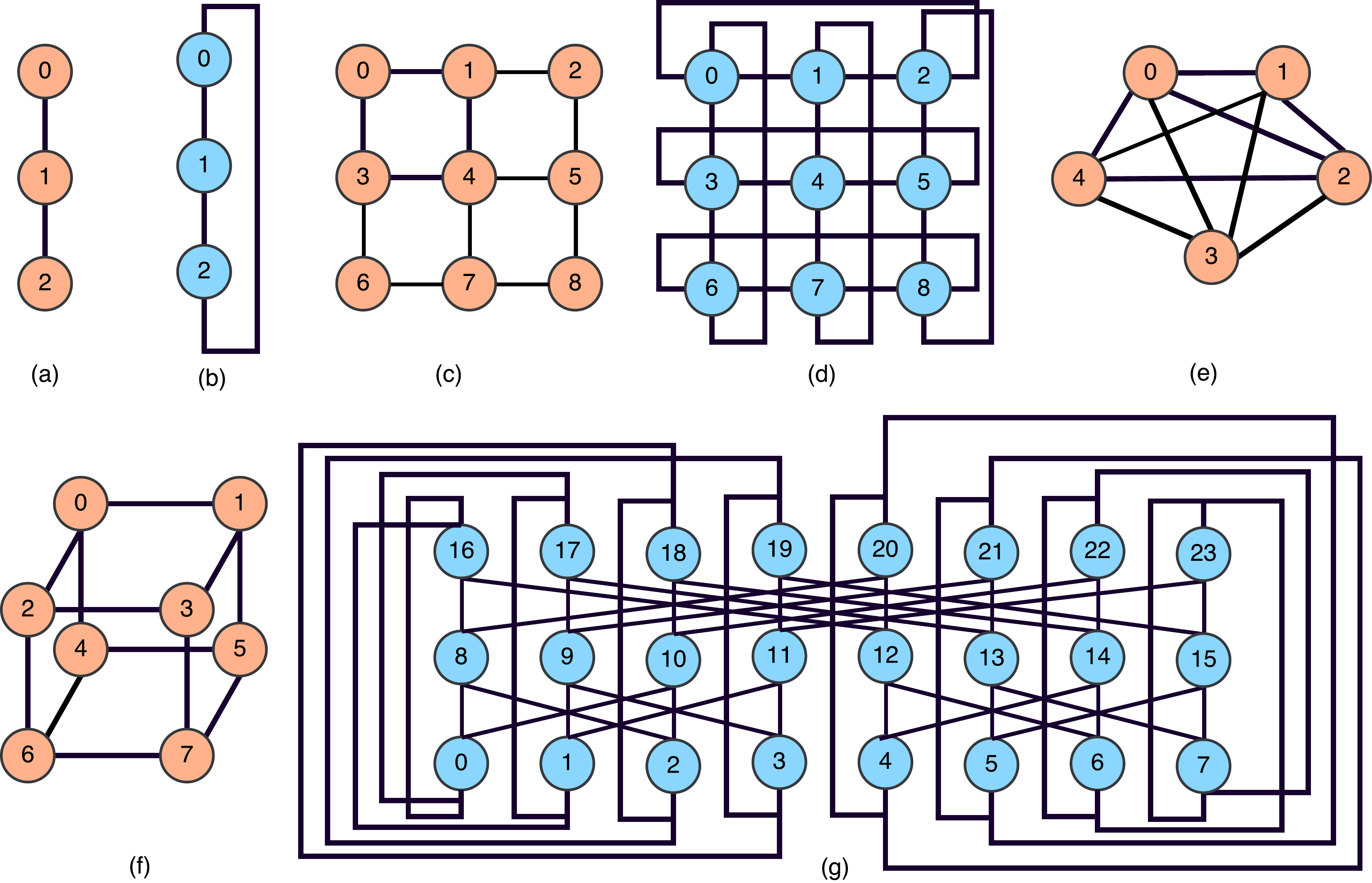}
 \caption{\em Topology (a) 1D-nearest neighbor (b) Cycle  (c) 2D-Mesh (d) Torus \mbox{(e) Fully connected graph} (f) 3D-Grid (g) Cyclic butterfly network}
 \label{fig:topo}
\end{figure}

\begin{definition}
 A \textbf{configuration} $C_t$ is the set of ordered tuples $(q_i, v)$, which indicates that in cycle $t$, qubit $q_i$,  is at location $v$, $ 1\le i \le n$ and $v \in T_\mathcal{V}$.
 Configuration $C_0$ represents the initial configuration.
\end{definition}

\subsection{Problem statement}\label{subsec:problem}
\noindent We now define three variants the nearest-neighbor optimization problem of quantum circuits for arbitrary topologies and also present
the relation between the variants.

\head{Problem $P_1$:} Given an initial configuration $C$ of $n$-inputs, an interaction $I$ and a topology graph $T$, the objective is to determine the series of swap gates needed
to transform the location of the qubits from configuration $C$ such that all qubit pairs in interaction I are nearest-neighbors and the delay due to insertion
of swap gates is minimum. 

    \vspace{\baselineskip}
\head{Problem $P_2$:} Given an initial configuration $C$ of $n$-inputs, a series of interactions $I_1, I_2, \ldots, I_k$ and a topology graph $T$, the objective is to determine the series of swap gates needed to transform the location of the qubits from configuration $C$ such that all qubits pairs in interaction $I_1$ are nearest-neighbor, and then again  
location of qubits are transformed to be nearest neighbors for $I_2$ and so on, till interaction $I_k$ is met and the delay due to insertion
of swap gates is minimum for the overall problem.  

    \vspace{\baselineskip}
\head{Problem $P_3$:}  Given an initial configuration $C$ of $n$-inputs, a series of levels $L_1, L_2, \ldots, L_k$ and a topology graph $T$, the objective is to determine the series of swap gates needed to transform the location of  the qubits from configuration $C$ such that all qubits pairs in interaction $I_1$~(corresponding on level $L_1$) are nearest-neighbor, and then again  
location of qubits are transformed to be nearest neighbors for $I_2$~(corresponding on level $L_2$) and so on, till interaction $I_k$~(corresponding on level $L_k$) is met and the combined delay of swap gates and gates present in the actual circuit is minimum. 

    \vspace{\baselineskip}
The Problem formulation $P_1$ has been popularly used for showing effectiveness of various topologies
to realize arbitrary permutations. Table~\ref{table:bounds} shows the depth and space requirements for realization of arbitrary permutations on various topologies. 
\begin{table}[]
\centering
\caption{Depth $D$ and Space $S$ complexity of realizing arbitrary permutations using a given topology.}
\label{table:bounds}
\begin{tabular}{llll}
\hline
\textbf{Topology}        & \textbf{Degree} & $\mathbf{D}$    & $\mathbf{S}$ \\ \hline
Fully Connected Graph    & $n$-1           & 1             & 1          \\
1D nearest-neighbor~\cite{hirata2011efficient}     & 2               & 2$n$-3        & 1          \\
2D nearest-neighbor~\cite{beals2013efficient}     & 4               & $O$($\sqrt n$)  & 1          \\
Cyclic butterfly network~\cite{brierley_butterfly} & 4               & 6log $n$      & 2          \\
Hypercube~\cite{beals2013efficient}                & log $n$         & $O($log$^2n)$ & 1          \\ \hline
\end{tabular}
\end{table}

Problem $P_2$  with $k$ = 1 is equivalent to Problem $P_1$. Therefore, finding an optimal solution for $P_2$ with $k$ = 1 is equivalent to solving $P_1$.  Problem $P_2$ does not consider the scheduling of the swap gates in parallel to quantum gates present in the original circuit, if possible. $P_2$ transforms the qubit locations on the topology graph such that the interactions needed to execute a level in quantum circuit is met.  Problem $P_3$ addresses this issue and considers the quantum gates as well and can find the optimal solution with minimum delay. 

\begin{theorem}\label{theo:delay}
 For a topology graph $T$ and a quantum circuit $C$ with $k+1$~levels, the delay $d_I$ of solution $S_I$ obtained by optimally solving problem $P_2$
 is at most $k$-cycles more than the delay $d_O$ of optimal solution $S_O$ of problem $P_3$ i.e. $d_I - d_O \le k$.
\end{theorem}
\begin{proof}
 Consider an  initial configuration and a quantum circuit two levels. Let us assume that the delay of solution be $d_I$ and $d_O$ be the optimal solution.
 Optimal solution for $P_3$ would have been able to insert additional gates at only one level $L_0$, which was not considered by $P_2$. If $d_I - d_O > 1$, this would imply that $d_I$ is not the optimal solution  for $P_2$, since there exists a solution to solve $P_2$ with $d_O+1$ delay which is a contradiction. This idea can be extended for any number of gates
 to derive Theorem~\ref{theo:delay}.
\end{proof}

\noindent It is possible to split the circuits into equal size blocks, with $b$-levels in each block, except the last block which might have less than $b$ levels. Fig.~\ref{fig:qc_block} shows the
blocks with size $b$=2, with the last block having a single interaction. Each block can solved using $P_2$ or $P_3$ to make the qubits nearest-neighbors 
and the output configuration of the solution is used as input configuration for the next block. 
For a quantum circuit with $k+1$-levels and $b \ge k$,
\begin{itemize}
 \item Optimal solution with minimum delay $d_O$ can be determined using $P_3$.  
 \item A bounded delay solution with delay $d_I$ can be determined using $P_2$ such that $d_I - d_O \le k$.
\end{itemize}
Various suboptimal solutions can be obtained using $b < k$, using both $P_2$ and $P_3$. Choosing a small block size $b$ makes it easier to solve each sub-problem
and therefore it becomes feasible to solve the nearest neighbor technology mapping problem for circuits with large number of gates. Corresponding to circuit in Fig.~\ref{fig:qc}, 
the 1D-nearest neighbor compliant circuit, obtained using problem formulation $P_2$ and $P_3$ for block size $b=4$, is shown in Fig.~\ref{fig:nnexample}~(a) and Fig.~\ref{fig:nnexample}~(b) respectively.
 
\begin{figure}[t!]
    \centering
    \subfloat[\em $P_2$ Solution with $b$ = 4]{\includegraphics[height= 1.4in]{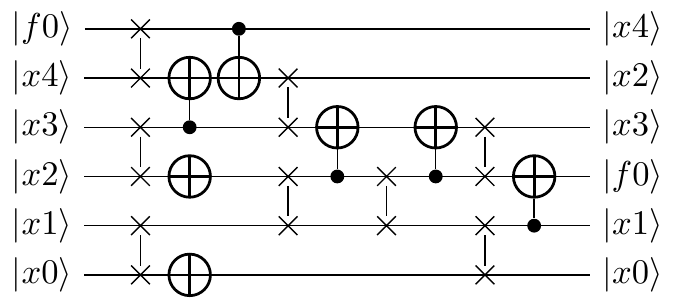}}\\
\subfloat[\em $P_3$ Solution with $b$ = 4]{\includegraphics[height= 1.4in]{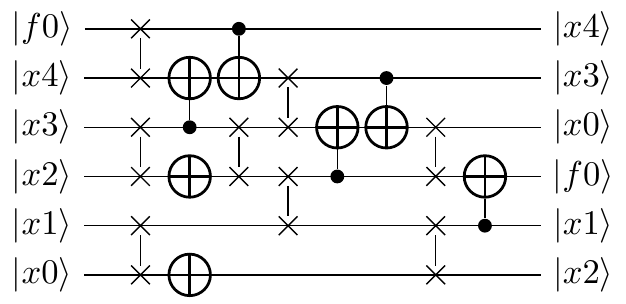}}
\caption{\em 1D-Nearest neighbor optimization solution}
\label{fig:nnexample}
\end{figure}
\section{Methodology}\label{sec:method}
\noindent In this section, we initially present an ILP formulation for the problem $P_2$. Description of the variables used in the 
formulation is presented summarily in Table~\ref{var:p2}. Thereafter, we present the modified ILP for problem $P_3$. 
\begin{table}
\centering
\caption{Parameters/constants used in ILP}
\label{table:param}
{\small
\begin{tabular}{cl}\bottomrule
\textbf{Param/const.} & \textbf{Description}  \\ \midrule
$G$  & Toplogy graph \\ 
$C$ & Input/start configuration \\ 
$n$ & Number of inputs \\ 
$k+1$ & Number of levels \\ 
$L_i$ & Number of qubit interaction pairs in level $i$ \\
$T$ & Maximum number of cycles used for the problem \\ \toprule
\end{tabular}
}
\end{table}
\begin{table}
\centering
\caption{Variable description used in ILP}
\label{var:p2}
{\small
\begin{tabular}{|c|p{5cm}|c|p{5cm}|}\hline
\textbf{Var.} & \textbf{Description}  &  \textbf{Var.} & \textbf{Description}\\\hline
$delay$ & Delay due to insertion of swap gates &$c_{v,q,t}$  & 1 indicates qubit $q$ will move to new location $v$ in cycle $t$ \\ \hline
$m_{i,t}$ & 0 indicates Interaction $i$ met in cycle $t$ &$a_{i,t}$ & 1 indicates gates in Level $i$ are scheduled in cycle $t$ \\ \hline
$n_{p,q,t}$ & 1 indicates qubit $p$ and $q$ are NN in cycle $t$ &$eb_{I_i,t}$ & 1 indicates interaction $I_i$ has been met in cycle $t$ and gates of level $i$ can be placed in the current or following cycles. \\ \hline
$p_{(p,v),(q,w),t}$ & 1 indicates qubit $p$ is in location $v$ and $q$ is in location $w$ in cycle $t$ &$b_{q,t}$ & 1 indicates qubit $q$ cannot be involved in a swap in cycle $t$ \\ \hline
$x_{v,p,t}$ & 1 indicates qubit $p$ is in location $v$ in cycle $t$ &$b_{v,q,t}$ & 1 indicates qubit $q$ in location $v$ cannot be involved in a swap in cycle $t$ \\ \hline
$u_{v,q,t}$ & 1 indicates qubit $q$ will remain in location $v$ in cycle $t$ &$sb_{m,n,t}$ & 1 indicates swap is not permitted between locations $m$ and $n$ in cycle $t$ \\ \hline
\end{tabular}
}
\end{table}

\subsection{ILP formulation for $P_2$}\label{subsec:p2}
\head{Objective function:}
\begin{align}\text{Minimize } &delay \\
\sum_{t=0}^T m_{k,t} - delay  &= 0
 \end{align}
 
\head{Chronological interaction constraints:} If an interaction is met in cycle $t$, then the status should not change to not met after
that cycle. In addition, interaction $i$ must be met before ${i-1}^{th}$  interaction is met.
\begin{align}
m_{i,t+1} - m_{i,t} &\ge 0 & 0 \le t \le T-1, 0 \le i \le k \\  
m_{i+1,t} -  m_{i,t} &\ge 0 &0 \le t \le T, 0 \le i \le k-1
\end{align}

\head{Successful interaction constraints:} An interaction is met if all the qubit pairs in the interaction are nearest neighbors. If an interaction has been met in
cycle $t$, then in all cycles $t' > t$, the qubit positions do not matter any longer.
\begin{align}
L_i.m_{i,t} + (\sum_{(p,q) \in I_i} n_{p,q,t}) + (\sum_{t'=0}^{t-1}L_i.(1 - m_{i,t'})) &\ge L_i &0\le t \le T
\end{align}

\head{Nearest neighbor constraints:} Two qubits $p$ and $q$ are nearest neighbors if the qubits are in two locations $v$ and $w$ respectively or in $w$ and $v$ respectively, such that $(v,w) \in G_{\mathcal{E}}$.
\begin{align}
p_{(p,v),(q,v),t} &= x_{v,p,t} \wedge x_{w,q,t} &(p,q) \in I, (v,w) \in  G_{\mathcal{E}} \\
p_{(p,w),(q,v),t} &= x_{w,p,t} \wedge x_{v,q,t}; &(p,q) \in I, (v,w) \in  G_{\mathcal{E}} \\
n_{p,q,t} &= \vee_{(v,w) \in G_{\mathcal{E}}} (p_{(p,v),(q,w),t} \vee p_{(p,w),(q,v),t}) & (p,q) \in I
\end{align}

\head{Qubit position update constraints:} A qubit $q$ is at location $v$ in cycle $t+1$ if it was in location $v$ in cycle $t$ and there were no swaps performed involving the location $v$ or
if $q$ was in a location $w$ which is nearest neighbor with $v$ and a swap was performed between $v$ and $w$.
\begin{align}
 u_{v,q,t+1} &= (\wedge_{(v,w) \in G_{\mathcal{E}}} (1 - s_{v,w,t}))\wedge x_{v,q,t};  \\
 c_{v,q,t+1} &= \vee_{(v,w) \in G_{\mathcal{E}}} s_{v,w,t}~\wedge ~x_{w,q,t}\\
 x_{v,q,t+1} &= u_{v,q,t+1} ~\vee~ c_{v,q,t+1}
\end{align}

\head{Qubit location and swap constraints:} A qubit $q$ can be at exactly one position in any given cycle. In a given cycle, a location can be involved in atmost one swap.
\begin{align}
 \sum_{v \in G_{\mathcal{V}}} x_{v,q,t} &= 1; &0 \le t \le T, q \in Q \\
 \sum_{(v,w) \in G_{\mathcal{E}}} s_{v,w,t} &\le 1; &0 \le t \le T, v \in G_{\mathcal{V}}
\end{align} 
\head{Initialization constraints:} A qubit $q$ is at location $v$ in cycle 0, based on input configuration $C$.
\begin{align}
 x_{v,q,0} &= 1; &(v,q) \in C
\end{align}
\noindent This concludes the description of the ILP formulation for problem $P_2$. The following subsection presents the modifications needed in the ILP for optimally solving $P_3$.

\subsection{ILP formulation for $P_3$}\label{subsec:p3}
\head{Objective function:}
\begin{equation}\text{Minimize }\sum_{i=0}^k\sum_{t = 0}^T t. a_{i,t} \end{equation}

\head{Level scheduling constraints:} Each level can be scheduled/activated exactly once.
\begin{align}
 \sum_{t=0}^{T} a_{i,t} &= 1; &0 \le i \le k
\end{align}
\noindent Only one level can be activated per time step.
\begin{align}
 \sum_{i=0}^{k} a_{i,t} &= 1; &0 \le t \le T
\end{align}
\noindent Activation for a level $i$ can happen only if corresponding interaction $i$ is met.
\begin{align}
 a_{i,t} + m_{i,t} &\le 1; &0 \le t \le T, 0 \le i \le k
\end{align}

\head{Swap blocking constraints:}
If an interaction $i'$ is met and all the gates in any Level $i$ such that $(i < i')$ have been scheduled, then swaps involving the qubits in interaction $i$ cannot be performed and interaction $i'$
is blocked till Level $i$ has been scheduled. Qubit involved in an interaction $i$ cannot be swapped in the cycle, when the Level $i$ is scheduled.
\begin{align}
eb_{i',t} &= a_{i,t} \wedge (1 - m_{i',t}); & 0 \le i \le k-1, i+1 \le i' \le k, 0 \le t \le T \\ 
b_{q,t} &= \vee_{i} (a_{i,t} \vee eb_{i,t}); &\forall i~ \exists ~q \in I_i, 0 \le t \le T \\
 b_{v,q,t} &= b_{q,t} \wedge x_{v,q,t} &0 \le t \le T \\
 sb_{m,n,t} &= \vee_q (b_{m,q,t} \vee b_{n,q,t}); &\forall q \in Q , 0 \le t \le T 
\end{align}

\noindent In addition to these constraints, {\em Chronological interaction constraints, Successful interaction constraints, Nearest neighbor constraints, Qubit position update constraints,
Qubit location and swap constraints} and {\em Initialization constraints} presented in ILP formulation for $P_2$ are applicable to $P_3$. This completes the description of the ILP formulation of $P_3$.

\section{Experimental Results}\label{sec:exp}
\noindent In this section, we present the benchmarking results for multiple quantum circuits from~\cite{revlib} for various topologies. We used Gurobi~\cite{gurobi} 
as ILP solver. For all the block sizes, we set TIME\_LIMIT parameter of  Gurobi to 600 seconds to limit the time of execution of the solver, except for solving 
full circuit optimization for which we set TIME\_LIMIT to 7200. We set the number of threads parameter in Gurobi to 8. For the experiments, we used 64-bit Ubuntu 14.04 running on Intel(R) Xeon(R) CPU E5-1650 v2@3.50GHz with 15.6 GB RAM.
\begin{table}[h]
\def\tabcolsep{3pt}
\centering
\vspace{0.1cm}
\caption{Realisation of all configurations of 4-qubits for 1D-topology}
\label{table:perm}
\begin{tabular}{|c|cc|c|cc|c|cc|}\bottomrule
\textbf{Config.}  & \textbf{\#}S & \textbf{D} & \textbf{Config.}  & \textbf{\#S} & \textbf{D} & \textbf{Config.}  & \textbf{\#S} & \textbf{D} \\ \hline
a b c d &0& 0&b c a d &2& 2&c d a b &2& 1\\
a b d c &1& 1&b c d a &3& 3&c d b a &1& 1\\
a c b d &1& 1&b d a c &3& 2&d a b c &3& 3\\
a c d b &2& 2&b d c a &2& 2&d a c b &2& 2\\
a d b c &2& 2&c a b d &2& 2&d b a c &2& 2\\
a d c b &3& 3&c a d b &3& 2&d b c a &1& 1\\
b a c d &1& 1&c b a d &3& 3&d c a b &1& 1\\
b a d c &2& 1&c b d a &2& 2&d c b a &0& 0\\\toprule
\end{tabular} 
\end{table}

Table~\ref{table:perm} demonstrates realization of all possible configurations of 4-variables for 1D topology. The initial configuration is assumed to [a,b,c,d]. \#S and D is the number of swap gates required and the corresponding delay to realise the target configuration respectively. This 
table has been obtained using Problem formulation $P_2$ with $k$=1. We would like to highlight that configuration [a b c d] and [d c b a] are identical since for both the configuration the pair of nearest-neighbor variables is same. 
\begin{table}[h]
\def\tabcolsep{3pt}
\centering
\caption{Benchmarking results for various block size for 1D-topology}
\label{table:1D}
\begin{tabular}{|lrrrc|rr|rr|rr|rr|rr|r}
\bottomrule
  \textbf{Benchmark}   &  \textbf{\#Var}  &  \textbf{\#Gates}  & \textbf{\#L}  &  \textbf{Tech.}  & \multicolumn{2}{c|}{\textbf{b=1}}  &  \multicolumn{2}{c|}{\textbf{b=2}}  &  \multicolumn{2}{c|}{\textbf{b=4}}  & 
  \multicolumn{2}{c|}{\textbf{b=8}}   &  \multicolumn{2}{c|}{\textbf{b=16}} \\
   &   &   &   &   & \#S  &  D  &  \#S  &  D &  \#S  &  D  &  \#S  &  D  &  \#S  &  D  \\ 
   \hline
 3\_17\_14 & 3 & 6 & 6 & $P_2$ & 3 & 9 & 3 & 8 & 3 & 8 & 3 & 8 & 3 & 8 \\  
 &  &  &  & $P_3$ & 3 & 9 & 3 & 8 & 3 & 8 & 3 & 8 & 3 & 8 \\ 
4gt11-v1\_85 & 5 & 4 & 3 & $P_2$ & 5 & 7 & 5 & 7 & 8 & 7 & 8 & 7 & 8 & 7 \\ 
 &  &  &  & $P_3$ & 5 & 7 & 5 & 7 & 7 & 7 & 7 & 7 & 7 & 7 \\ 
4mod5-v1\_25 & 5 & 4 & 3 & $P_2$ & 3 & 5 & 3 & 5 & 3 & 5 & 3 & 5 & 3 & 5 \\ 
 &  &  &  & $P_3$ & 3 & 5 & 3 & 5 & 4 & 5 & 4 & 5 & 4 & 5 \\ 
alu-bdd\_288 & 7 & 9 & 8 & $P_2$ & 22 & 19 & 19 & 17 & --- & --- & --- & --- & --- & --- \\ 
 &  &  &  & $P_3$ & 22 & 19 & 26 & 18 & --- & --- & --- & --- & --- & --- \\ 
ex-1\_166 & 3 & 4 & 4 & $P_2$ & 1 & 5 & 1 & 5 & 1 & 5 & 1 & 5 & 1 & 5 \\ 
 &  &  &  & $P_3$ & 1 & 5 & 1 & 5 & 1 & 5 & 1 & 5 & 1 & 5 \\ 
ex1\_226 & 6 & 7 & 5 & $P_2$ & 7 & 9 & 7 & 9 & 8 & 9 & 8 & 9 & 8 & 9 \\  
 &  &  &  & $P_3$ & 8 & 9 & 8 & 9 & 8 & 7 & 8 & 7 & 8 & 7 \\ 
fredkin\_7 & 3 & 1 & 1 & $P_2$ & 0 & 1 & 0 & 1 & 0 & 1 & 0 & 1 & 0 & 1 \\  
 &  &  &  & $P_3$ & 0 & 1 & 0 & 1 & 0 & 1 & 0 & 1 & 0 & 1 \\ 
graycode6\_48 & 6 & 5 & 5 & $P_2$ & 0 & 5 & 0 & 5 & 0 & 5 & 0 & 5 & 0 & 5 \\  
 &  &  &  & $P_3$ & 0 & 5 & 0 & 5 & 0 & 5 & 2 & 5 & 2 & 5 \\ 
ham3\_103 & 3 & 4 & 4 & $P_2$ & 4 & 7 & 3 & 6 & 3 & 6 & 3 & 6 & 3 & 6 \\ 
 &  &  &  & $P_3$ & 4 & 7 & 3 & 6 & 3 & 6 & 3 & 6 & 3 & 6 \\ 
mod5d2\_70 & 5 & 8 & 7 & $P_2$ & 6 & 12 & 6 &  & 8 & 12 & 9 & 10 & 10 & 10 \\  
 &  &  &  & $P_3$ & 6 & 12 & 8 & 10 & 6 & 12 & 10 & 10 & 10 & 10 \\ 
one-two-three-v3\_101 & 5 & 8 & 7 & $P_2$ & 11 & 13 & 11 & 13 & 10 & 13 & 10 & 13 & 10 & 13 \\  
 &  &  &  & $P_3$ & 12 & 15 & 14 & 14 & 9 & 12 & 9 & 12 & 9 & 12 \\ 
peres\_9 & 3 & 2 & 2 & $P_2$ & 2 & 4 & 2 & 4 & 2 & 4 & 2 & 4 & 2 & 4 \\ 
 &  &  &  & $P_3$ & 2 & 4 & 2 & 4 & 2 & 4 & 2 & 4 & 2 & 4 \\ 
rd32\_272 & 5 & 6 & 5 & $P_2$ & 12 & 11 & 9 & 11 & 9 & 11 & 9 & 11 & 9 & 11 \\ 
 &  &  &  & $P_3$ & 10 & 11 & 12 & 11 & 9 & 11 & 9 & 11 & 9 & 11 \\ 
toffoli\_double\_4 & 4 & 2 & 2 & $P_2$ & 3 & 4 & 3 & 4 & 3 & 4 & 3 & 4 & 3 & 4 \\ 
 &  &  &  & $P_3$ & 3 & 4 & 3 & 4 & 3 & 4 & 3 & 4 & 3 & 4 \\ 
xor5\_254 & 6 & 7 & 5 & $P_2$ & 7 & 9 & 7 & 9 & 6 & 9 & 9 & 8 & 9 & 8 \\ 
 &  &  &  & $P_3$ & 8 & 9 & 8 & 8 & 8 & 7 & 8 & 7 & 8 & 7 \\ \toprule
\end{tabular}
\end{table}

Table~\ref{table:1D} presents the results of 1D-Nearest neighbors for multiple block size $b=$\{1, 2, 4, 8, 16\}. The column Tech. indicates whether the solution for a benchmark is obtained using the problem formulations $P_2$ or  $P_3$. Using a large block size is expected to reduce overall circuit delay, since the optimization solver can search a larger solution space to obtain optimal solution in that space instead of hitting a locally optimal solution. For circuit {\em xor5\_254} , the delay for block side $b=1$ is 9 while that 
with block size $b=4$ is 7. On the other hand, by using a smaller block, it is possible to obtain
a feasible solution within the  time limits specified for the solver, since the solver has to solve a smaller instance of the formulated ILP. For example, solutions could not be obtained for $b \ge 4$ within the specified time limits for circuit {\em alu-bdd\_288}. It should be noted that for all block sizes $b < L$, where $L$ is the number of levels in the circuit, the overall circuit is not guaranteed to have least delay, even when using formulation $P_3$ since combining the optimal solutions of the subproblems does not guarantee globally optimal solution.

\begin{table}[h]
\def\tabcolsep{3pt}
\centering
\vspace{0.1cm}
\caption{Benchmarking results for {\em 4gt10-v1\_81} using $P_3$ formulation, $w=1$}
\label{table:b2}
\begin{tabular}{|c|cc|cc|cc|cc|cc|cc|cc|} \bottomrule
& \textbf{\#G} & D & \multicolumn{2}{c|}{\textbf{1D}}  &  \multicolumn{2}{c|}{\textbf{Cycle}}  &  \multicolumn{2}{c|}{\textbf{2D-Mesh}}  & 
  \multicolumn{2}{c|}{\textbf{Torus}}   &  \multicolumn{2}{c|}{\textbf{3D-Grid}}  &  \multicolumn{2}{c|}{\textbf{CBN}}  \\ 
& & &\#S & D & \#S & D & \#S & D & \#S & D & \#S & D & \#S & D  \\\hline
Original & 6 & 6 & \multicolumn{2}{c|}{NF} &  \multicolumn{2}{c|}{NF} & 11  & 11 & 5 & 9 & 7 & 11 & ---&--- \\  
Decomposed & 12 & 12 & 25&26 &  14 & 21  & 10&22  &  6&15 & 15&23 & --- & ---\\  
\toprule
\end{tabular} 
\end{table}
We demonstrate the impact of topology on feasibility of nearest neighbor mapping for a given circuit.
For this purpose, we used the circuit $4gt10-v1\_81$ shown in Fig.~\ref{fig:4gt}. The circuit has a Toffoli gate with 3-control lines. This cannot be mapped using 1D-NN or cycle topology because a qubit can have at most two-neighbors in 1D or cycle topology. However, for other topologies, the mapping is feasible and the results using formulation $P_3$ are presented in Table~\ref{table:b2}. In order to make the nearest neighbor mapping feasible, the Toffoli gate with $n$-controls can be decomposed into a sequence of 2-control Toffoli gates~\cite{barenco1995elementary,maslov2008quantum}. We used the RC-Viewer+ tool~\cite{rcviewer} to decompose the circuit as shown in Fig.~\ref{fig:4gtdecom}, followed by problem formulation $P_3$ to solve the nearest neighbor mapping problem for the same. As evident from the results, the decomposed circuit is now feasible to be 
mapped to 1D-NN and cycle topologies. For the other topologies, the mapping of the decomposed circuit 
has worser delay compared to the mapping of the original circuit, due the higher number of levels in the decomposed circuit. 
%
\begin{figure}[h]
    \centering
    \begin{minipage}{1.8in}
        \centering
        \includegraphics[height=0.8in]{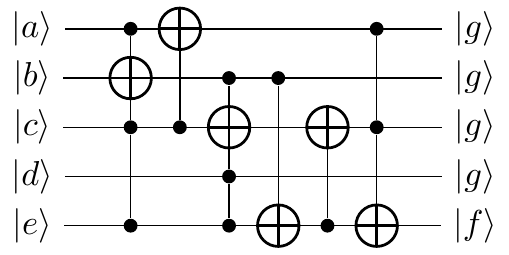}
        \caption{\em Benchmark circuit }
        \label{fig:4gt}
    \end{minipage}%
    \begin{minipage}{2in}
        \centering
        \includegraphics[height=0.8in]{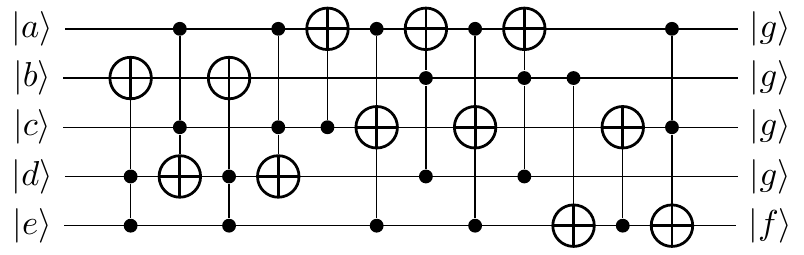}
        \caption{\em Decomposed Circuit}
        \label{fig:4gtdecom}
    \end{minipage}
\end{figure}

For the first time, we report results for multiple topologies for various standard benchmark quantum circuits in Table~\ref{table:topo}. For each circuit, we consider the smallest topology
graph with number of nodes greater than or equal to number of variables in the circuit. We have considered an arbitrary initial placement of the qubits on the topology graph. 
As expected, topologies with greater number of edges have lower delay. For example, the delay obtained for $cycle$ topology is less than that for $1D$ topology. Multiple benchmarks for the 3D-Grid and cyclic butterfly network~(CBN) did not complete execution within the specified time limit, due to the relatively large size of the topology graphs.

\begin{table}[h]
\def\tabcolsep{3pt}
\centering
\caption{Benchmarking results for entire circuit}
\label{table:topo}
\scalebox{0.9}{
\begin{tabular}{|lrrrc|rr|rr|rr|rr|rr|rr|}
\bottomrule
  \textbf{Benchmark}   &  \textbf{\#Var}  &  \textbf{\#Gates}  & \textbf{\#L}  &  \textbf{Tech.}  & \multicolumn{2}{c|}{\textbf{1D}}  &  \multicolumn{2}{c|}{\textbf{Cycle}}  &  \multicolumn{2}{c|}{\textbf{2D-Mesh}}  & 
  \multicolumn{2}{c|}{\textbf{Torus}}   &  \multicolumn{2}{c|}{\textbf{3D-Grid}}  &  \multicolumn{2}{c|}{\textbf{CBN}}\\
   &   &   &   &   & \#S  &  D  &  \#S  &  D &  \#S  &  D  &  \#S  &  D  &  \#S  &  D  &  \#S  &  D \\
\hline
 3\_17\_14.real & 3 & 6 & 6 & $P_2$ & 3 & 7 & 0  &  6 & 10 & 7 & 0 & 6 & 7 & 8 & 0 & 6 \\ 
 &  &  &  & $P_3$ & 3 & 7 & 0 & 6 & 14 & 7 & 0 & 6 & --- & --- & 0 & 6 \\  
4gt11-v1\_85.real & 5 & 4 & 3 & $P_2$ & 8 & 7 & 3 & 5 & 1 & 4 & 3 & 4 & 0 & 3 & 1 & 4 \\ 
 &  &  &  & $P_3$ & 7 & 7 & 4 & 5 & 4 & 4 & 9 & 4 & --- & --- & --- & ---- \\ 
4mod5-v1\_25.real & 5 & 4 & 3 & $P_2$ & 3 & 5 & 2 & 4 & 3 & 5 & 2 & 4 &--- & --- & --- & ----\\ 
 &  &  &  & $P_3$ & 4 & 5 & 3 & 4 & 3 & 5 & 6 & 4   &--- & --- & --- & ---- \\ 
alu-bdd\_288.real & 7 & 9 & 8 & $P_2$ & 18 & 15 & --- & --- & --- & --- & --- & --- & --- &---  &---  &---  \\ 
 &  &  &  & $P_3$ & --- & --- & --- & ---  & 16 & 10 & 15 & 9 & --- & --- & --- & ---- \\ 
ex-1\_166.real & 3 & 4 & 4 & $P_2$ & 1 & 4 & 0 & 3 & 1 & 4 & 0 & 3 & 1 & 4 & 0 & 3 \\ 
 &  &  &  & $P_3$ & 1 & 5 & 0 & 4 & 4 & 5 & 2 & 4 & 6 & 5 & 8 & 4 \\  
ex1\_226.real & 6 & 7 & 5 & $P_2$ & 8 & 9 & 7 & 8 & 4 & 7 & 2 & 6 & --- & --- & --- & ---  \\ 
 &  &  &  & $P_3$ & 8 & 7 & 8 & 7 & 12 & 6 & 9 & 5 &  --- & --- & --- & ---- \\ 
fredkin\_7.real & 3 & 1 & 1 & $P_2$ & 0 & 1 & 0 & 1 & 0 & 1 & 0 & 1 & 0 & 1 & 0 & 1 \\ 
 &  &  &  & $P_3$ & 0 & 1 & 0 & 1 & 0 & 1 & 0 & 1 & 0 & 1 & 0 & 1 \\  
graycode6\_48.real & 6 & 5 & 5 & $P_2$ & 0 & 5 & 0 & 5 & 5 & 6 & 4 & 6 &---  & --- & 0 & 5 \\ 
 &  &  &  & $P_3$ & 2 & 5 & 0 & 5 & 7 & 5 & 8 & 5 &  --- & --- & --- & ---- \\ 
ham3\_103.real & 3 & 4 & 4 & $P_2$ & 3 & 6 & 1 & 4 & 3 & 6 & 1 & 4 & 3 & 6 & 1 & 4 \\ 
 &  &  &  & $P_3$ & 3 & 6 & 1 & 4 & 8 & 6 & 2 & 4 & 11 & 6 & 21 & 4 \\  
mod5mils\_71.real & 5 & 5 & 5 & $P_2$ & 4 & 7 & 4 & 7 & 3 & 5 & 2 & 5 & --- & --- & ---- & ---  \\ 
 &  &  &  & $P_3$ & 6 & 7 & 4 & 6 & 5 & 5 & 6 & 5 &  --- & --- & --- & ---- \\ 
one-two-three-v3\_101.real & 5 & 8 & 7 & $P_2$ & 10 & 13 & 7 & 11 & 10 & 11 & 6 & 9 &  &  &  &  \\ 
 &  &  &  & $P_3$ & 9 & 12 & 7 & 11 & --- & --- & 13 & 8 &  --- & --- & --- & ---- \\ 
peres\_9.real & 3 & 2 & 2 & $P_2$ & 2 & 4 & 0 & 2 & 3 & 4 & 0 & 2 & 14 & 4 & 0 & 2 \\ 
 &  &  &  & $P_3$ & 2 & 4 & 0 & 2 & 9 & 4 & 2 & 2 & 2 & 4 & 0 & 2 \\  
rd32\_272.real & 5 & 6 & 5 & $P_2$ & 9 & 11 & 4 & 8 & 5 & 7 & 7 & 7 &  &  &  &  \\ 
 &  &  &  & $P_3$ & 9 & 11 & 6 & 8 & 12 & 7 & 8 & 6 & --- & --- & --- & ---- \\ 
toffoli\_double\_4.real & 4 & 2 & 2 & $P_2$ & 3 & 4 & 1 & 3 & 4 & 4 & 2 & 3 & 14 & 3 & 1 & 3 \\ 
 &  &  &  & $P_3$ & 3 & 4 & 1 & 3 & 4 & 3 & 5 & 3 & 4 & 3 & 4 & 3 \\  \toprule 
\end{tabular} 
}
\end{table}

Direct comparison of our method to obtain nearest-neighbor compliant circuits with existing works could not be performed for primarily three reasons. 
The existing works~\cite{lye2015determining, kole2016heuristic,shafaei2013optimization,saeedi2011synthesis}  focus on determining linear nearest neighbors~(LNN), with the objective of reducing number of swap gates. Our proposed method is 
for obtaining the LNN circuits with minimal depth which is contrary to the goal of reducing swap gate count. Secondly, the initial placement of the qubit is assumed to be given as input to the problem, but 
other works consider this as part of the optimization. Finally, most of the existing works decompose the gates into two qubit gates~\cite{wille_aspdac16,shafaei2014qubit,kole2016heuristic,saeedi2011synthesis}.
In our work, we used unmodified circuits from RevLib~\cite{revlib}. 
For reference of the readers, we provide a brief summary of the existing results in terms of number of swap gates against the solution of our proposed methodology using problem formulation $P_3$ with block size $b=4$, for the decomposed circuits in Table~\ref{table:existing}. Due to non-availability of the depth of the transformed circuits, we cannot compare the performance of our method against the existing works.

\begin{table}
\centering
\caption{Comparison with existing works on LNN}
\label{table:existing}
\begin{tabular}{|lll|cccc|}\bottomrule
 \textbf{Benchmark}   & \textbf{\#Var} & \textbf{\#Gates} & $P_3$($b=4$) & N=4\cite{kole2016heuristic} & \cite{saeedi2011synthesis} & \cite{shafaei2013optimization}\\\hline
3\_7\_13 & 3 & 14	& 7 & 6 & 6 & 4 \\ 
4\_49\_17 & 7 & 32 	&15 & 15 & 20 & 12 \\ 
4gt10-v1\_81 & 5 &	 36	& 33 & 22 & 30 & 20 \\ 
4gt11\_84 & 5 & 7		& 5 & 5 & 3 & 1 \\ 
4gt13-v1\_93 & 5 &	 17	& 18 & 10 & 11 & 6 \\ 
4gt5\_75 & 5 & 	22	& 25 & 15 & 17 & 12 \\ 
4mod5-v1\_23 & 5 & 24	& 22 & 13 & 16 & 9 \\ 
alu-v4\_36 & 5 &	32	& 26 & 22 & 23 & 18 \\ 
hwb4\_52 & 4 & 23		& 13 	& 9 & 14 & 10 \\ 
ham7\_104 & 7 &	87	& 140 & 83 & 84 & 68 \\ 
mod5adder\_128 & 6 & 87	& 94 & 65 & 85 & 51 \\  \toprule 
 \end{tabular}
\end{table}

%
\section{Conclusion}\label{sec:conclusion}
\noindent In this paper, we addressed the problem of nearest-neighbor optimization
for a given quantum circuit, an arbitrary topology graph and an initial configuration specifying
the location of qubits in the topology graph. We formulated the problem using two ILP variants --- one of the 
variant for obtaining the optimal solution and a simpler variant that can obtain a bounded solution.
In addition, our problem formulation allows the optimization to be performed as
a large set of small optimizations or a smaller set of larger optimization problems, by setting appropriate
block sizes. We demonstrated the effectiveness of our approach by running it on a set of benchmark circuits. 
Further research can be undertaken to solve the same problem for arbitrary topologies by 
heuristic based approaches that would allow scaling for larger circuits.

\bibliographystyle{IEEEtran}
\bibliography{./lit/lit_header,./lit/lit_misc,./lit/lit_mymisc,./lit/lit_myrev,./lit/lit_others,./lit/lit_othersrev,./lit/lit_physics,./lit/lit_lye,./lit/lit_anupam}

\end{document}